\documentclass{l4dc2025}

\title[Data-Driven Optimal Control via Koopman Generator]{Data-driven optimal control of unknown nonlinear dynamical systems using the  Koopman operator}
\usepackage{times}
\usepackage{booktabs,makecell}
\usepackage{appendix}
\usepackage{amsmath,amssymb,amsfonts}
\usepackage{empheq}
\usepackage{graphicx}
\usepackage{xcolor}
\usepackage{graphicx}
\usepackage{times}
\usepackage[T1]{fontenc}
\usepackage{dcolumn}
\usepackage{bm}
\usepackage{colortbl}
\usepackage{graphicx}
\usepackage{textcomp}
\usepackage{multicol}
\usepackage{multirow}
\usepackage{chngpage}
\usepackage{array}
\usepackage{enumerate}
\usepackage{algpseudocode}  
\usepackage{mathrsfs}
\usepackage{makecell}
\usepackage[most]{tcolorbox}
\usepackage{balance}
\newtheorem{assumption}{Assumption}
\newcommand{\by}{{\bm y}}
\newcommand{\bM}{{\mathcal{M}}}

\newcommand{\bR}{\mathbb{R}}
\newcommand{\bN}{\mathbb{N}}
\newcommand{\bC}{\mathbb{C}}

\newcommand{\bx}{{\bm{x}}}
\newcommand{\bz}{{\bm{z}}}
\newcommand{\bu}{{\bm{u}}}

\newcommand{\bU}{\mathcal{U}}

\newcommand{\hL}{\mathcal{L}} 
\newcommand{\hF}{\mathcal{F}}

\newcommand{\sspan}{\mathrm{span}}
\newcommand{\argmin}{\mathrm{argmin}}

\renewcommand{\rm}[1]{\mathrm{#1}}

\author{%
 \Name{Zhexuan Zeng} \Email{xuanxuan@hust.edu.cn}\\
 \addr Department of Automatic Control, Huazhong University of Science and Technology, Wuhan, China.
	\AND
 \Name{Ruikun Zhou} \Email{ruikun.zhou@uwaterloo.ca}\\
 \addr Department of Applied Mathematics,
        University of Waterloo, Waterloo, Ontario N2L 3G1, Canada.
 \AND
 \Name{Yiming Meng} \Email{ymmeng@illinois.edu}\\
 \addr Coordinated Science Laboratory, University of Illinois Urbana-Champaign,
Urbana, IL 61801, USA.
 \AND
  \Name{Jun Liu} \Email{j.liu@uwaterloo.ca}\\
 \addr Department of Applied Mathematics,
        University of Waterloo, Waterloo, Ontario N2L 3G1, Canada.
}

\begin{document}
	
	\maketitle
	
	\begin{abstract}%
		Nonlinear optimal control is vital for numerous applications but remains challenging for unknown systems due to the difficulties in accurately modelling dynamics and handling computational demands, particularly in high-dimensional settings. This work develops a theoretically certifiable framework that integrates a modified Koopman operator approach with model-based reinforcement learning to address these challenges. By relaxing the requirements on observable functions, our method incorporates nonlinear terms involving both states and control inputs, significantly enhancing system identification accuracy. Moreover, by leveraging the power of neural networks to solve partial differential equations (PDEs), our approach is able to achieving stabilizing control for high-dimensional dynamical systems, up to 9-dimensional. The learned value function and control laws are proven to converge to those of the true system at each iteration. Additionally, the accumulated cost of the learned control closely approximates that of the true system, with errors ranging from $10^{-5}$ to $10^{-3}$.
	\end{abstract}
	
	
	\section{Introduction}
    
	A central problem in control engineering is nonlinear optimal control, which 
    has broad applications across various fields, including autonomous vehicle navigation, satellite and spacecraft control, and robotic manipulators. 
    
    A natural approach to pursue optimal control for continuous-time nonlinear dynamical systems is first linearizing the system at each state, representing the nonlinear dynamics as a state-dependent linear system. This allows the control law to be derived by solving state-dependent Riccati equation \citep{farsi2022piecewise}. However, this approach typically yields only a sub-optimal controller. An alternative method for solving the optimal control problem involves addressing the Hamilton-Jacobi-Bellman (HJB) equation. Since the HJB equation is a nonlinear partial differential equation that is notoriously difficult to solve directly, most research focuses on obtaining approximate solutions indirectly through policy iteration techniques \citep{leake1967construction,saridis1979approximation,beard1995improving,jiang2017robust}. Originating from the optimal control of Markov decision processes \citep{bellman1957dynamic,howard1960dynamic}, policy iteration begins with an initial stabilizing control and iteratively improves the closed-loop performance through two key steps: policy evaluation and policy improvement. Specifically, policy evaluation involves solving the Generalized Hamilton-Jacobi-Bellman (GHJB) equation, a linear partial differential equation that is generally more tractable than the HJB equation. For low-dimensional problems, Galerkin approximations have demonstrated their effectiveness in providing accurate solutions to the HJB equation with arbitrary precision \citep{beard1997galerkin,beard1998approximate}. To overcome the curse of dimensionality in high-dimensional systems, neural networks are increasingly employed to approximate the solution to the GHJB equation. These networks ensure convergence to the true solution at each iteration, leveraging their ability to approximate complex functions and scale efficiently with problem size \citep{meng2024physics,zhou2024physics}.
	
However, solving the GHJB equation requires complete knowledge of system dynamics, which is often unavailable in practice. To address this challenge, various methods based on adaptive dynamic programming (ADP) have been developed to approximate the value function and control laws directly from online measurements, such as, \citep{jiang2012computational,vrabie2009neural,jiang2014robust}. 
These model-free methods are particularly advantageous because they can leverage advanced techniques, such as deep learning, to address high-dimensional state space problems. Despite their flexibility, model-free methods often suffer from a lack of theoretical guarantees and involve high implementation complexity. In contrast, model-based methods can use established control theories—such as stability analysis, performance optimization—to design control laws with rigorous guarantees. However, their effectiveness heavily depends on the accuracy of the identified model, which can be challenging to achieve in high-dimensional scenarios.    

In recent years, the Koopman operator \citep{Koopman1931hamiltonian} has gained significant attention due to its ability to provide a linear representation of nonlinear systems within a function space. Through numerical algorithms such as Dynamic Mode Decomposition (DMD) \citep{schmid2010dynamic} and Extended DMD \citep{williams2015data,korda2018convergence}, Koopman operator-based methods have proven highly effective for system identification \citep{mauroy2019koopman} and for analyzing identifiability in relation to sampling frequency \citep{zeng2022sampling, zeng2024sampling, zeng2024generalized} in autonomous systems. To address control problems, many studies have explored representing nonlinear systems with inputs using the Koopman operator \citep{korda2018linear, mauroy2020koopman}. Although theoretically feasible by treating inputs as augmented 
states, practical implementation of this framework remains challenging, as it often necessitates neglecting certain terms such that the system becomes completely linear with respect to both the lifted state and the input. This limitation reduces both the accuracy of the identified system and the effectiveness of the resulting controllers, especially for high-dimensional systems.

To address the aforementioned limitations, the main contributions of this paper are as follows:
\begin{itemize}
\item We develop a theoretically certifiable framework, integrating a modified Koopman operator approach with model-based reinforcement learning, for control system identification and deriving optimal control policies for unknown nonlinear systems.
\item We improve the Koopman operator-based identification accuracy by relaxing the requirements on observable functions, allowing accurate recovery of nonlinear control-affine terms.
\item We enhance the scalability of computing optimal control directly from data by leveraging the power of neural networks to solve PDEs. We demonstrate the effectiveness of our method through four example systems with state dimensions ranging from 2 to 9 and input dimensions up to 4.
\end{itemize}
	
	\section{Problem formulation}\label{sec:prob}
	We consider a control-affine dynamical system:\begin{equation}\label{sys}
		\dot{\bx} = f(\bx)+g(\bx)\bu,
	\end{equation} where $\bx\in\bR^n$ denotes the state vector, $\bu\in\bR^m$ denotes the input, $f:\bR^n\to\bR^n$ is a continuously differentiable vector field, and $g:\bR^n\to\bR^{n\times m}$ is a smooth function. We also assume that $\bm 0$ is an equilibrium point of \eqref{sys}, i.e., $f(\bm 0) = \bm 0$. Subject to the control $\bu$, the unique solution starting from $\bx_0$ is denoted by $S^t(\bx_0,\bu)$.
	
	We define the space of admissible controls of \eqref{sys} as follows. 
	\begin{definition}[Admissible control]
		A feedback control $\bu = \kappa(\bx)$ is admissible on $\Omega\subset\bR^n,$ where $\bm 0 \in \Omega$, if the following conditions are satisfied: (1) $\kappa$ is Lipschitz continuous on $\Omega$; (2) $\kappa (\bm 0)  = \bm 0$; (3) $\bu = \kappa(\bx)$ is a stabilizing control, i.e., $\lim_{t\to\infty}|S^t(\bx_0,\bu)| = 0$ for $\forall \bx_0\in \Omega$. We denote the space of admissible control as $\bU(\Omega)$.
	\end{definition}
	
	We are interested in finding the optimal control $\kappa^*(\bx) \in \bU(\Omega)$ from data, in the case of infinite interval $t\in[0,\infty)$. Specifically, we introduce the function $L(\bx,\bu)=Q(\bx)+\|\bu\|_R^2,$ where $Q:\bR^n\to\bR$ is a positive definite function, and $\|\bu\|_R^2 = \bu^T R \bu$ given  some symmetric and positive definite $R: \bR^n\to\bR^{m\times m}$.  The associated cost is commonly defined as follows:
	\begin{equation}
		J(\bx,\bu) = \int_{0}^{\infty} L(S^t(\bx,\bu),\bu(t)){\rm d} t
	\end{equation}
	The optimal control is denoted as $\bu^* = \kappa^*(\bx)$ such that \begin{equation}
		J(\bx,\bu^*) = \inf_{\kappa\in\bU} J(\bx,\kappa(\bx)),
	\end{equation} and the value function is defined as $V(\bx):=J(\bx,\bu^*)$. 
    Intuitively, the value function describes the system's infimum energy loss over the state space for all possible control inputs, thereby providing a foundation for deriving the optimal control $\bu^* = \kappa^*(\bx).$ 
    Specifically, the optimal control is derived by minimizing the Hamiltonian: \begin{equation}
    \kappa^*(\bx) := \argmin_{\bu\in\bU}
 \{L(\bx,\bu)+ DV(\bx)\cdot f(\bx,\bu)\}=-\frac{1}{2}R^{-1}g^T(\bx)(DV(\bx))^T.
\end{equation}
	
This paper aims to systematically explore theoretically certifiable method for computing optimal control from data in unknown high-dimensional nonlinear systems. 
	
	\section{Preliminaries of exact policy iteration and Koopman operator theory}\label{sec:pre}
	\subsection{Exact policy iteration}
We begin by reviewing the policy iteration method for systems with known dynamics.	The value function that we aim to find is generally a viscosity solution to the HJB equation, i.e.,\begin{equation}
		H(\bx,DV(\bx)) = 0,
	\end{equation} where $
		H(\bx, p) := \sup_{\bu\in\bR^m}-G(\bx,\bu,p),$ and $
		G(\bx,\bu,p) = L(\bx,\bu)+p(f(\bx)+g(\bx)\bu).$
	However, solving and analyzing this equation is a complex task.
    The policy iteration method 
	assumes that $V \in C^1(\Omega)$ and seeks $C^1$ solutions $V_i$ to the GHJB equation $G(\bx, \bu_i, DV_i(\bx)) = 0$ for each iteration $i \in \{0, 1, \cdots \}$, specifically, given an initial input $\bu_0=\kappa_0(\bx)$ that stabilizes the system, the policy iteration method performs policy evaluation and policy improvement iteratively:\begin{itemize}
		\item[1] The $i$-th policy evaluation: Given a policy $\bu_i = \kappa_i(\bx)$, solving the GHJB equation below to compute the value function $V_i(\bx)$ at $x\in\Omega\setminus\{\bm 0\}$, and we set $V_i(\bm 0) = 0$.
		\begin{equation}\label{GHJB}
			G(\bx,\kappa_i(\bx),DV_i(\bx)) := L(\bx,\kappa_i(\bx))+DV_i(\bx)(f(\bx)+g(\bx)\kappa_i(\bx)) = 0.
		\end{equation}
		\item[2] The $i$-th policy improvement: Given the value function $V_i(\bx)$, solving the GHJB to update the policy:
		\begin{equation} \label{policy}
			\begin{aligned}
				\kappa_i(\bx) = \left\{\begin{matrix}
					&-\frac{1}{2}R^{-1}g^{T}(\bx)(DV_i(\bx))^T, \quad \bx\neq\bm 0;\\
					&\bm 0,\quad \bx = \bm 0.
				\end{matrix} \right.
			\end{aligned}
		\end{equation}
	\end{itemize}
	
	Assuming that $V\in C^1(\Omega)$, the convergence value function $V_\infty$ is expected to solve the HJB equation and $\bu_i\to\bu^*$ at least pointwise \citep[Theorem 3.1.4]{jiang2017robust}. 

    \subsection{Koopman operator theory}
To solve  \eqref{GHJB} for unknown systems, we first identify $f(\bx)$ and $g(\bx)$.
Koopman operator provides an alternative perspective to analyze and learn nonlinear dynamical systems. Below we briefly introduce the Koopman operator theory. Let us first consider $\bu = \bm 0$. Then the dynamical system \eqref{sys} becomes:
	\begin{equation}\label{eq1}
		\begin{aligned}
			\dot{\bx}&=f (\bx).
		\end{aligned}
	\end{equation}
    The flow induced by this autonomous system is denoted as $S^t(\bx,\bm 0),t>0$, i.e., $\bx(t)=S^t(\bx(0),\bm 0)$. The Koopman operator $U^t:\mathcal F\to\mathcal F$ is a linear operator acting on the observable functions of the states, i.e., $g\in\hF:\bR^n\to\bC$, which is defined as \begin{equation}\label{defkoopman}
		U^t g(\bx) = g(S^t(\bx, \bm 0)).
	\end{equation} 
	The infinitesimal generator $\hL$ of the Koopman operator is defined as \begin{equation}
		\hL g = \lim_{t\to 0^+}\frac{1}{t}(U^t-I)g, ~g\in\mathcal D(\hL),
	\end{equation} where $\mathcal D(\hL)$ denotes the domain of $\hL$. The generator is also a linear operator. Assuming that observable functions $g\in\mathcal F$ are continuously differentiable with compact support, we have $
		\hL = f\cdot \nabla$.
	
	Due to its linearity and rich theoretical support, the Koopman operator theory enjoys wide popularity in nonlinear system identification and control. Despite its success, accurately representing a nonlinear system with input as a linear input-output system remains challenging, as it often requires disregarding certain terms that describe how control actions evolve in the observation space.
	
	\section{Description of the method} \label{sec:method}
	The main idea of the method is to first identify the nonlinear dynamical system with control using the generator. 
    Then we solve its
	optimal control problems by iterative procedure utilizing random neural approximations. In the following, we describe the steps in detail.
	
	\subsection{Identification of control-affine system}
	The nonlinear system is equivalently described as an infinite-dimensional linear system driven by the Koopman operator. In practice, we lift and embed the original system into a high-dimensional function space, then approximate the Koopman operator and its generator using a linear, matrix-like operator defined on the same function space domain.
	
	\subsubsection{Lifting of the dynamical system}
	Theoretically, we view the input $\bu$ of \eqref{sys} as an external state of the dynamical system and assume that it remains unchanged during the sampling time. Then we have the extended system: 
        \begin{equation}\label{eq:sample_control_id}
        \begin{aligned}
		\dot{\bx} &= f(\bx)+g(\bx)\bu,\\
		\dot{\bu} &= 0,
	\end{aligned} 
    \end{equation}
    and the corresponding extended state $[\bx,\bu]^T$. For simplicity of the notation, we also denote the flow of the extended system as $S^t:\bM\to\bM$ with the invariant set $\bM$ of the extended states, where $\bM \subseteq \mathbb{R}^{n+m}$.
	
	To accurately characterize the original system, we recover the generator within an observable space where the basis functions $\{\varphi_i(\bx,\bu)\}$ include coupling terms between $\bx,\bu$. This approach avoids constructing a high-dimensional linear input-output system of the form $\dot{\bz} = A\bz+B\bu$, where $\bz = \varphi(\bx)$ and $\varphi$ is a vector-valued function consisting of multiple scalar observation functions that depend solely on $\bx$. In this work, we select the polynomial observable functions of $(\bx,\by)$. Furthermore, to recover $f(\bx)$ and $g(\bx)$ in the system, we restrict the total degree of each control input $u_i,i=1,\ldots,m$ to $1$ within the observable functions, i.e., $
		\hF_N = \sspan_{i=1,\ldots,N}\{\varphi_i(\bx,\bu)\},$ where $\varphi_i(\bx,\bu) = \prod_{j=1}^{n}x_j^{p_j}\prod_{l=1}^{m} u_l^{q_l}$, $\sum_{l=1}^{m}q_l \le 1,$ with $x_j$ being the $j$-th component of $\bx$, $u_l$ being the $l$-th component of $\bu$, $q_l, p_j\in\bN$.
	
	\subsection{Identification of the generator}
	Typically, the Koopman operator is first identified from data, and the generator is then obtained by taking the logarithm of the Koopman operator \citep{mauroy2019koopman}. However, this indirect approach requires the chosen observable space to be invariant under both the Koopman operator and its generator, which is a condition that is challenging to meet in practice. Consequently, this method often introduces greater approximation error compared to directly identifying the generator.
	
	To avoid indirect approximation error, we identify the generator based on the Yosida approximation that utilizes the resolvent operator of the Koopman operator. 

The Yosida approximation $L_\lambda$, defined as \begin{equation}\label{E: yosida}
		L_\lambda \varphi_i(\bx)= \lambda^2\int_{0}^{\infty}e^{-\lambda t}U^t\varphi_i(\bx){\rm d}t - \lambda \varphi_i(\bx),
	\end{equation} converges to the generator $\hL$ on $C^1(\bM)$ as $\lambda\to\infty$ in a strong sense \citep[Theorem 3.3]{meng2024koopman}, i.e., $L_\lambda h\to \hL h$ for any observable function $h\in C^1(\bM)$. To ensure numerical tractability, we further approximate \eqref{E: yosida} for each observable function $\varphi_i$ using a truncated integral over a fixed finite-time horizon $[0, T_\text{max}]$, as follows:
	\begin{equation}\label{E: yosida_trunc}
		L_{\lambda,T_\text{max}} \varphi_i(\bx) = \lambda^2\int_{0}^{T_\text{max}}e^{-\lambda t}U^t\varphi_i(\bx){\rm d}t - \lambda \varphi_i(\bx).
	\end{equation} 
    
    Given the choice of polynomial observables $\{\varphi_i\}$, the overall approximation error of $\hL\varphi_i$ (for each $i$) using $L_{\lambda,T_\text{max}}\varphi_i$ is of the order $\mathcal{O}(e^{-\lambda T_\text{max}})$, as stated in \cite[Theorem 4.2]{meng2024koopman}.  Notably, for large $\lambda$, truncating the integral at any $T_\text{max}$ has a non-dominant impact on the approximation accuracy. This error bound allows us to directly obtain the value of $\hL\varphi_i$ using the evaluation of \eqref{E: yosida_trunc}, and further adapt it with sampled data using numerical quadrature techniques.

    This allows us to construct two matrices as follows: \begin{equation}
		X = \left(\begin{matrix}
			\varphi_1(\bx_1) & \ldots & \varphi_N(\bx_1)\\
			\vdots & \ddots & \vdots\\
			\varphi_1(\bx_M) & \ldots & \varphi_N(\bx_M)\\
		\end{matrix}\right), Y =  \left(\begin{matrix}
			L_{\lambda,T_{max}}\varphi_1(\bx_1) & \ldots & L_{\lambda,T_{max}}\varphi_N(\bx_1)\\
			\vdots & \ddots & \vdots\\
			L_{\lambda,T_{max}}\varphi_1(\bx_M) & \ldots & L_{\lambda,T_{max}}\varphi_N(\bx_M)\\
		\end{matrix}\right),
	\end{equation} where $\{\bx_i\}_{i=1}^M$ denote $M$ samples. 
	Then we compute the matrix representation of the generator as $\hat{L}_N = (X^T X)^{-1} X^TY.$

	\subsection{Recovery of nonlinear dynamical system}
	To recover $f(\bx)$ and $g(\bx)$ from the identified generator, we denote the $i$ as the index of the observable function such that $\varphi_i(\bx) = x_j,$ where $j=1,\ldots,n.$ Then we have \begin{equation}\label{identifiedfg}
		\hat{f}_j(\bx)+\hat{g}_j(\bx)\bu = \sum_{k=1}^{N}\varphi_k(\bx,\bu) [\hat{L}_N]_{k,i},
	\end{equation} where $[\hat{L}_N]_{k,j}$ denotes the $k$-th row, $j$-th column of $\hat{L}_N$. We can approximate $\hat{f}(\bx)$ and $\hat{g}(\bx)$ corresponding to observable functions $\varphi_k(\bx,\bu) = \prod_{j=1}^{n}x_j^{p_j}\prod_{l=1}^{m} u_l^{q_l}$ with $\sum_{l=1}^{m}q_l = 0$ and $\sum_{l=1}^{m}q_l = 1$, respectively.

	\subsection{Policy iteration via linear least squares}
	Based on the identified $\hat{f}(\bx)$ and $\hat{g}(\bx)$, we continue to solve the optimal control problem by employing policy evaluation \eqref{GHJB} and policy improvement \eqref{policy}. Specifically, in the $i$-th iteration, we solve the following equation:\begin{align}
		&L(\bx,\kappa_i(\bx))+DV_i(\bx)(\hat{f}(\bx)+\hat{g}(\bx)\kappa_i(\bx)) = 0,\label{PE}\\
		&\kappa_i(\bx) = \left\{\begin{matrix}
			-\frac{1}{2}R^{-1}\hat{g}^{T}(\bx)(DV_i(\bx))^T, ~\bx\neq\bm 0;\\
			\bm 0, \bx = \bm 0.
		\end{matrix} \right.
	\end{align}
	To solve \eqref{PE} and approximate the value function, we use random feature neural network functions, resulting in a neural solution of the form $
		\hat{V}(\bx) = \bm \beta^T\sigma(W\bx+\bm b),$ where $\bm \beta\in\bR^s, W \in \bR^{s\times n}, \bm b \in \bR^s$, and $\sigma: \bR\to\bR$ is an activation function applied element-wise. It follows that $
		D\hat{V}(\bx) = \bm \beta^T \text{ diag}(\sigma'(W\bx+\bm b))W.$ Noted that $W$ and $\bm b$ can be randomly chosen, which does not require training. Then the problem of solving \eqref{PE} transforms to the problem of finding the parameter $\bm \beta$. Due to the linear dependence of $D\hat{V}(\bx)$ on $\bm \beta$ and the linearity of \eqref{PE}, it results in a linear least squares optimization problem that can be solved efficiently and accurately.
	\section{Theoretical convergence}\label{sec:conv}
    We begin by proving the convergence of identified system. In the following, we use $L_{\lambda, T_{\rm max}, N}$ to denote $N$-dimensional approximation of $L_{\lambda, T_{\rm max}}$, $F_{\lambda, T_{\rm max}, N}(\bx,\bu)$ to denote the vector field recovered from $L_{\lambda, T_{\rm max}, N}$, and $F(\bx,\bu)$ to denote the original vector field.
	\begin{theorem}[Convergence of vector field]\label{fconvergence}
		As $\lambda\to\infty, T_{max}\to\infty, N\to\infty$ simultaneously, we have $F_{\lambda,T_{max},N}\to F$ uniformly on  $\bM$, where $\bM$ is a compact set in $\mathbb{R}^{n+m}$.
	\end{theorem}
	\begin{proof}
		For the $i$-th component of vector field, where $i=1,\ldots,n$, we have $F_i(\bx,\bu) = \hL \varphi_q(\bx,\bu), $ $F_{i;\lambda,T_{max},N}(\bx,\bu) = L_{\lambda,T_{max},N} \varphi_q(\bx,\bu),$ where $q$ is the index of the observable function such that $\varphi_q(\bx,\bu) = x_i$. It follows that 
        \begin{align*}
			\|F_i-F_{i;\lambda,T_{max},N}\|_{\infty}  &\le  \|(L-L_\lambda)\varphi_q\|_\infty + \|(L_\lambda - L_{\lambda,T_{max}})\varphi_q\|_\infty + \|(L_{\lambda,T_{max}} - L_{\lambda,T_{max},N})\varphi_q\|_\infty, 
		\end{align*} 
        where $\|\cdot\|_\infty$ denotes $\sup_{(\bx,\bu)\in \bM}\|\cdot\|$ with $\|\cdot\|$ being the 2-norm. 
        Based on Theorem 4.2, Theorem 3.3, and Corollary 4.6 in \cite{meng2024koopman}, we have $\|F_i-F_{i;\lambda,T_{max},N}\|_{\infty} \to 0$ as $\lambda\to\infty, T_{max}\to\infty, N\to\infty$ simultaneously.
	\end{proof}


    While Theorem \ref{fconvergence} guarantees theoretical convergence of $F_{\lambda,T_{max},N}$ to $F$, as $\lambda\to\infty, T_{max}\to\infty, N\to\infty$, the following assumption states that, with sufficiently many samples, the identified system $\hat F(\bx,\bu) = \hat f(\bx) + g(\bx)\bu$ should be close to $F_{\lambda,T_{max},N}$.

    \begin{assumption}\label{assum_data}
        For $\forall \theta>0,$ the initial conditions for (\ref{eq:sample_control_id}) can be sampled sufficiently densely in $\bM$ such that $\hat{F}(\bx,\bu) = \hat{f}(\bx)+\hat{g}(\bx)\bu$ identified from \eqref{identifiedfg} satisfies $\|F_{\lambda,T_{\rm max}, N}-\hat{F}\|_\infty \le \theta$.
    \end{assumption}

    Under Assumption \ref{assum_data} and based on Theorem \ref{fconvergence}, we can conclude that, for every $\theta>0$, there exist sufficiently large $\lambda,T_{\rm max}, N$ and sufficiently dense initial conditions such that \begin{align}\label{assum1_F}
        \|f(\bx) + g(\bx)\bu - \hat{f}(\bx) - \hat{g}(\bx)\bu\|\le \theta,\quad \forall (\bx,\bu)\in \bM.
    \end{align} 
    Without loss of generality, we assume that \eqref{assum1_F} holds for $\bu\in B = \left\{\|\bu\|\le 1\right\}$. Letting $\bu = \bm 0$, we have $\|f - \hat{f}\|_\infty\le \theta.$ It follows from the triangle inequality that $\|g(\bx)\bu -  \hat{g}(\bx)\bu\|\le 2\theta.$  With $\bu = \frac{g(\bx) - \hat{g}(\bx)}{\|g(\bx)-\hat{g}(\bx)\|}\in B$, this implies $\|g - \hat{g}\|_\infty\le 2\theta.$ In other words, for every $\theta>0$, there exist sufficiently large $\lambda,T_{\rm max}, N$ and sufficiently dense initial conditions such that $\|f - \hat{f}\|_\infty\le \theta, \|g-\hat{g}\|_\infty\le 2\theta$, which are essential requirements for the following analysis. 
    
    We expect each policy evaluation of the identified system to closely approximate that of the true system, ensuring that the algorithm ultimately produces a meaningful result. The following theorem establishes that, with each iteration, the value function and control derived from the identified system converge to those of the true system. The proof is provided in the arXiv version \citep{zeng2024optimalcontrol}. For brevity, we directly use the index $h = (\lambda,T_{max},N)$, and $f_h(\bx)+g_h(\bx)\bu$ to denote the identified vector field from data.

\begin{theorem}\label{conv_pi}
		Let $\Omega\subset\bM$ a compact invariant set for each $\dot{\bx} = F_h^{(i)}(\bx,\kappa_h^{(i-1)}(\bx)) = f_h(\bx)+g_h(\bx)\kappa_h^{(i-1)}(\bx)$ and $\dot{\bx} = F^{(i)}(\bx,\kappa^{(i-1)}(\bx)) = f(\bx)+g(\bx)\kappa^{(i-1)}(\bx)$. Assume $F_h^{(i)}$, $F^{(i)}\in C^1(\Omega)$, $L^{(i)}_h, L^{(i)}\in C^1(\Omega)$. Then, for $\forall \theta>0,$ there exists $\delta>0$ such that 
		if $\|\kappa_{h}^{(0)}-\kappa^{(0)}\|_\infty <\delta, \|f_h-f\|_\infty <\delta, \|g_h-g\|_\infty <\delta$
		, we have $\|V^{(i)}_h-V^{(i)}\|_\infty <\theta, \|\kappa_h^{(i)}-\kappa^{(i)}\|_\infty <\theta$ for all $\bx\in \Omega$.
	\end{theorem}
    
	\section{Numerical experiments}\label{sec:num}
	
	\subsection{Experimental setup}
	
	To demonstrate the performance of the proposed method, we learn the optimal control for the following systems: $2$-dimensional (inverted pendulum), $4$-dimensional (cartpole), $6$-dimensional (2D quadrotor), and $9$-dimensional systems (3D quadrotor). 
	
	In the identification step: For $2$-dimensional and $4$-dimensional systems, we select the space of polynomials $\varphi_i(\bx,\bu) = \prod_{j=1}^{n}x_j^{p_j}\prod_{l=1}^{m} u_l^{q_l}$ where $p_j\le p_{max}$ for $j=1,\ldots,n$. For $6$-dimensional and $9$-dimensional systems, we select the space of polynomials constrained by $\sum_{j=1}^{n}p_j\le p_{sum}$. To ensure that the identified system has an equilibrium at $\bm 0$, we exclude constant functions from the set of observable functions, i.e., $\sum_{j=1}^{n}p_j+\sum_{l=1}^{m}q_l > 0$. For each system, we collect data with the time horizon $t\in[0,1]$ and $100$Hz sampling frequency. The specific parameters and data details for identification and policy iteration steps are provided in Tables \ref{tbl:identify_para} and \ref{tbl:pi_para} below.

	\begin{table*}[!htb] 
		\caption{The detailed information of data and parameters for identification}
		\centering
		\begin{tabular}{cccc}
\toprule	
{Dynamical system}&{Domain}&{Polynomial order}& {Initial samples}
			
			\\
			\hline
			a) Inverted pendulum              & $(\bx,\bu)\in[-1,1]^3$     &      $p_{max} = 5$    & $1000$    \\ 
			{b) Cartpole}               & $(\bx,\bu)\in[-0.2,0.2]^5$   &      $p_{max} = 3$ & $3125$        \\ 
			{c) 2D quatroter}     & $(\bx,\bu)\in[-0.2,0.2]^8$  &      $p_{sum} = 3$   & $5000$       \\ 
			{d) 3D quadrotor}     & $(\bx,\bu)\in[-0.2,0.2]^{13}$  &      $p_{sum}=3$  & $10000$        \\ 
    \bottomrule
		\end{tabular}
		\label{tbl:identify_para}
	\end{table*}
	
	\begin{table*}[!htb]
		\caption{The detailed information of data and parameters for policy iteration}
		\centering
		\begin{tabular}{c c c c}
\toprule			{Dynamical system}&{Domain}&{Hidden unites (s)}& {Samples}
			
			\\
			\hline
			a) Inverted pendulum              & $\bx\in[-1,1]^2$     &      $200$    & $3000$    \\ 
			{b) Cartpole}               & $\bx\in[-0.1,0.1]^4$   &      $3200$ & $6000$        \\ 
			{c) 2D quadroter}               & $\bx\in[-0.1,0.1]^6$  &      $3200$   & $9000$       \\ 
			{d) 3D quadrotor}     & $\bx\in[-0.1,0.1]^{9}$  &      $3200$  & $12000$        \\ 
    \bottomrule
		\end{tabular}
		\label{tbl:pi_para}
	\end{table*}
	\subsection{Numerical results}    
    \emph{1) Identification performance.} To effectively solve the HJB equation, it is crucial to accurately estimate $f(\bx)$ and $g(\bx)$ over the considered set. To demonstrate the advantages of the proposed method, we compare the evaluation error of our approach (the resolvent-based model) with the logarithm-based model proposed by \cite{mauroy2019koopman} and the widely used lifted linear model in control frameworks \citep{korda2018linear}. The evaluation errors are calculated as $E_f = \sum_{i=1}^M \|f(\bx_i) - \hat{f}(\bx_i)\|_1/M, E_g = \sum_{i=1}^M \|g(\bx_i) - \hat{g}(\bx_i)\|_1/M,$ where $\|\cdot\|_1$ denotes the element-wise sum of absolute values. To ensure a fair comparison, it is important to emphasize that the basis functions used in the logarithm-based method are equivalent to those in our approach. Additionally, the resolvent-based identification method is employed to identify the lifted linear model, where the observable function exclude cross-terms involving $\bx$ and $\bu$. To minimize the influence of the number of basis functions on the comparison, we further increase the polynomial order in the lifted linear model, ensuring its basis functions are at least equal to or more comprehensive than those used in our method. 
    
    The comparison results are detailed in Table \ref{compa_sysid}. These results demonstrate that our method achieves evaluation errors for $f$ and $g$ that are reduced by one to two orders of magnitude compared to the other two approaches, which is particularly pronounced in high-dimensional systems. This improvement in accuracy is crucial for solving the HJB equation, as our attempts with the other two methods failed to learn a control law to stabilize the cartpole, 2D quadrotor, and 3D quadrotor.  In contrast, the control law and value function learned using our method are presented below.
    
    	\begin{table*}[!htb]
  \centering
    \caption{Comparison of evaluation error for ours (resolvent-based control-affine model), LAM (logarithm-based control-affine model), and RLM (resolvent-based lifted linear model). 
    }
  \begin{tabular}{cccccccccccc}
    \toprule \multicolumn{3}{c}{Inverted pendulum} & 
    \multicolumn{2}{c}{Cartpole} &
    \multicolumn{2}{c}{2D quadrotor} &\multicolumn{2}{c}{3D quadrotor} 
     \\
    \cmidrule(r){1-3} \cmidrule(r){4-5} \cmidrule(r){6-7}\cmidrule(r){8-9} 
     & $E_f$ & $E_g$ &$E_f$ & $E_g$& $E_f$ & $E_g$ & $E_f$ & $E_g$ \\
    \cmidrule(r){1-3} \cmidrule(r){4-5} \cmidrule(r){6-7}\cmidrule(r){8-9} 
  \textbf{Ours} & 3.7E-3 & \textbf{9.9E-3} & \textbf{3.5E-3}& \textbf{2.0E-3} &\textbf{8.6E-4} & \textbf{2.2E-3}  &  \textbf{1.7E-3} & \textbf{8.5E-3}   \\
  LAM & 2.9E-1 & 6.4E-1 & 7.8E-2 & 1.4E-2 & 1.4E-2 &2.1&  2.4E-2 & 4.3E-2 \\
 RLM & 2.8E-3 & 1.2E-2 & 1.8E-2 & 3.4E-2 & 9.5E-3& 1.1E-1  & 1.8E-2 & 2.2E-1 \\

    \bottomrule
  \end{tabular}
  \label{compa_sysid}
\end{table*}

    \emph{2) Control performance.}
We randomly choose $50$ initial conditions in the associated domain of Table \ref{tbl:identify_para}, and we simulate these trajectories using the true system and the control learned from the identified system. To illustrate the performance of the learned control, we compute the average $\hat{C}(t)$ of the accumulated costs $\hat{C}_i(t), t\in[0,10]$ for these $50$ trajectories. We also perform the simulation and compute the average of the accumulated cost $C(t)$ using the learned control from the true system. The error of the mean accumulated cost $|\hat{C}(t) - C(t)|$ and the trajectories are depicted in Fig. \ref{fig2}. These results demonstrate that, when the dynamical system is unknown, the accumulated cost of the optimal control input obtained by this method closely aligns with that of the optimal control learned from the true system, with errors ranging from $10^{-5}$ to $10^{-3}$. The optimal control input learned from data of this unknown system effectively stabilizes the trajectories of this true system. 

For comparison, ADP \citep[Chapter 3]{jiang2017robust} performs well in low-dimensional cases, such as the 2D pendulum. However, it faces significant challenges in learning stable controllers for relatively higher-dimensional systems such as cartpole. Consequently, our method demonstrates strong potential for addressing optimal control problems in high-dimensional systems.
	
		
	
	\begin{figure}[thpb]
		\centering
        		\subfigure[Cost error of inverted pendulum]{
			\label{Fig1.sub1}
			\includegraphics[width=.38\textwidth]{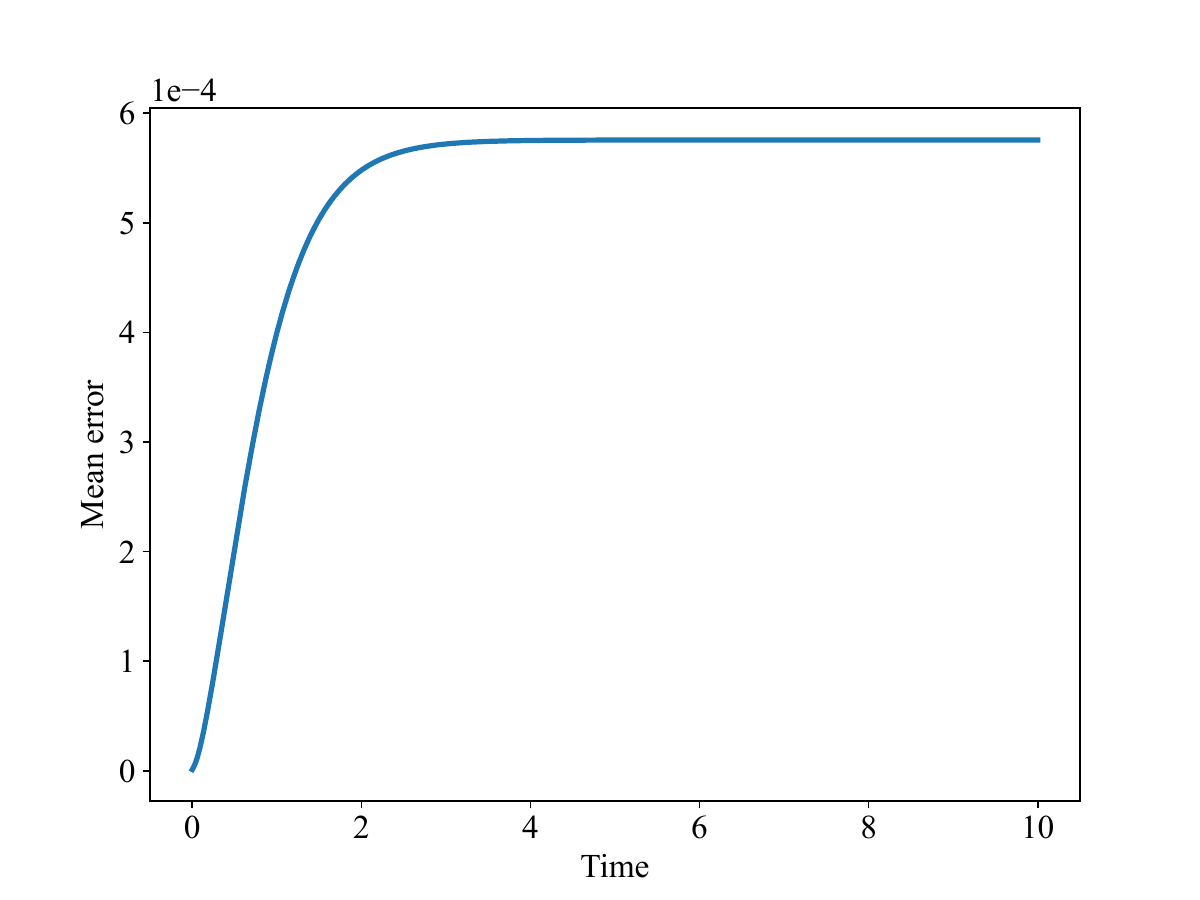}}
		\subfigure[Trajectories for inverted pendulum]{
			\label{Fig1.sub2}
			\includegraphics[width=.38\textwidth]{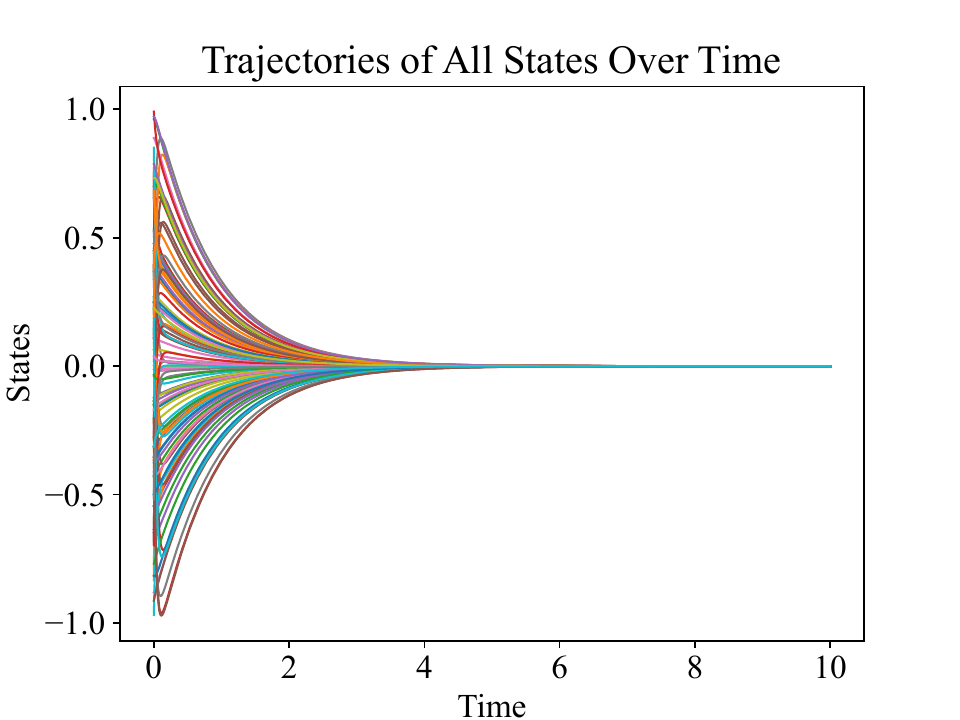}}\\
		\subfigure[Cost error of cartpole]{
			\label{Fig1.sub5}
			\includegraphics[width=.38\textwidth]{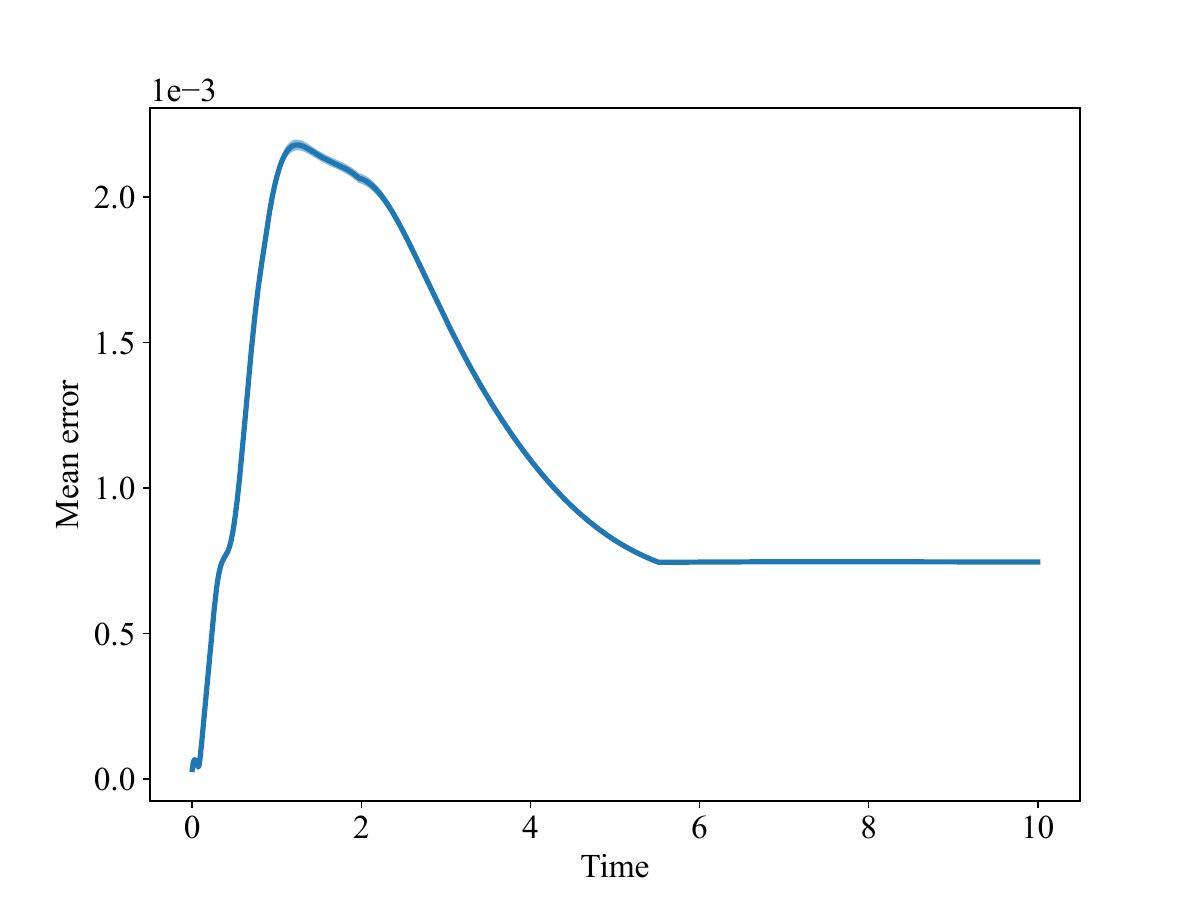}}
		\subfigure[Trajectories for cartpole]{
			\label{Fig1.sub6}
			\includegraphics[width=.38\textwidth]{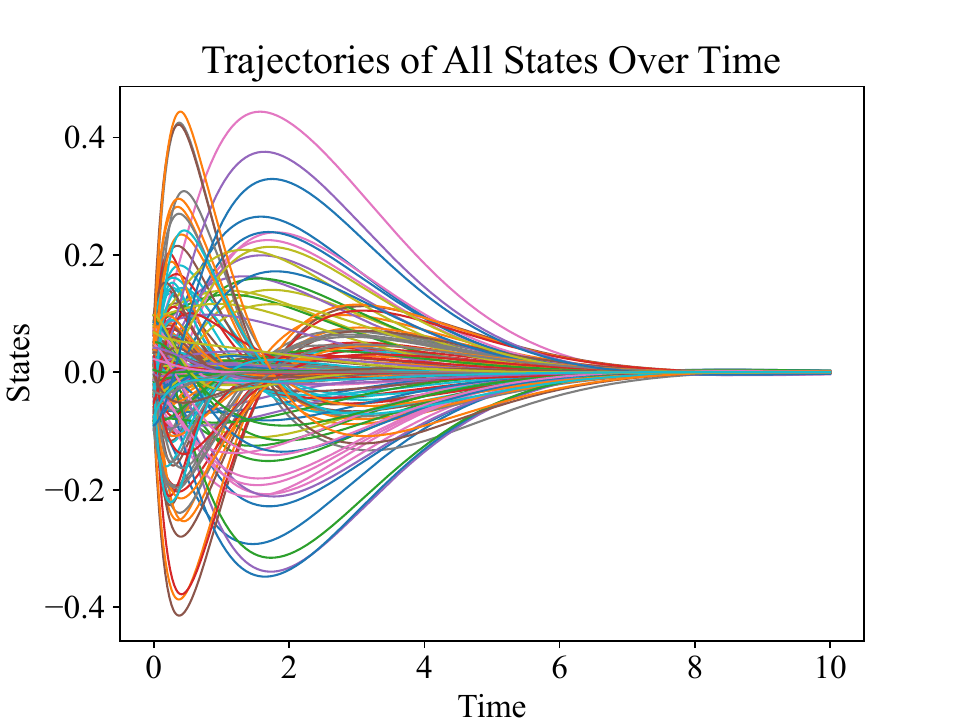}
		}
		\subfigure[Cost error of 2D quadrotor]{
			\label{Fig2.sub1}
			\includegraphics[width=.38\textwidth]{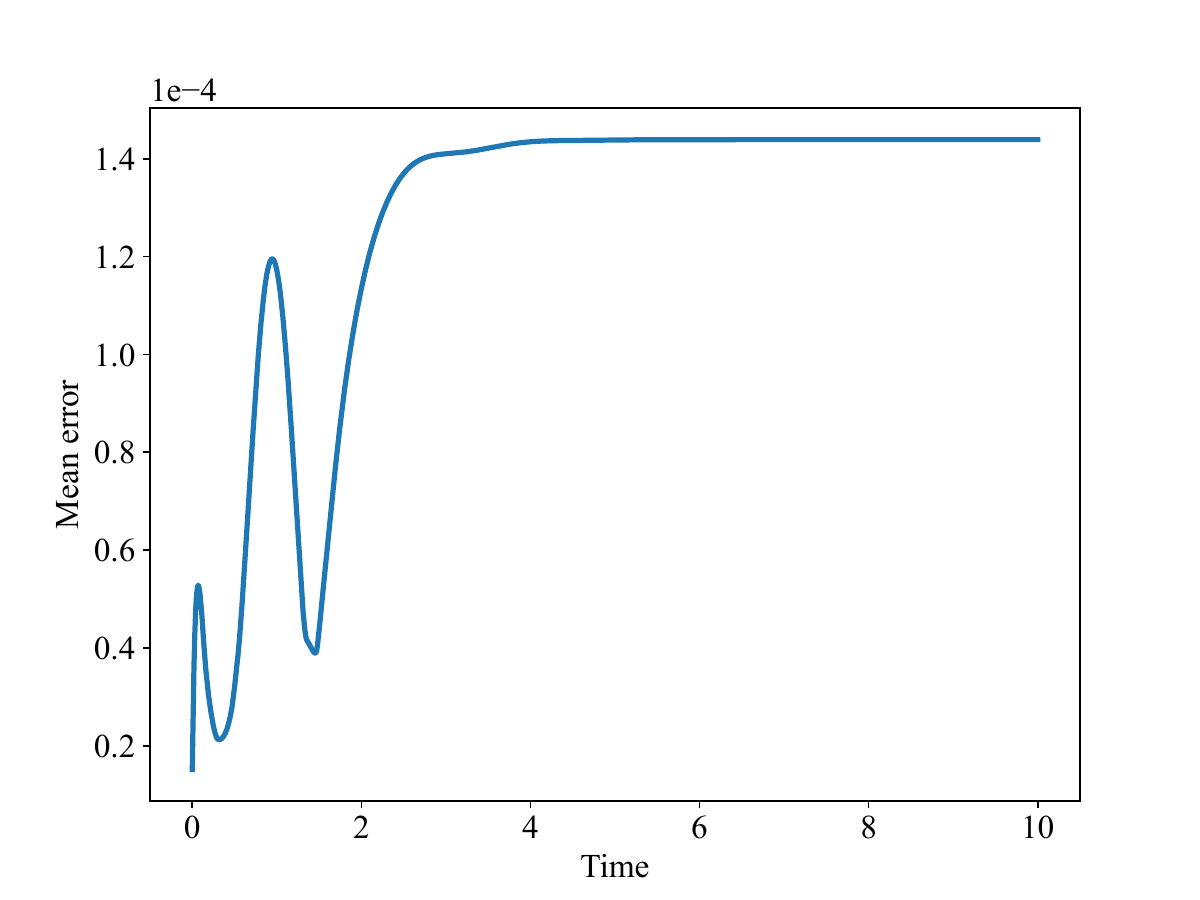}}
		\subfigure[Trajectories for 2D quadrotor]{
			\label{Fig2.sub2}
			\includegraphics[width=.38\textwidth]{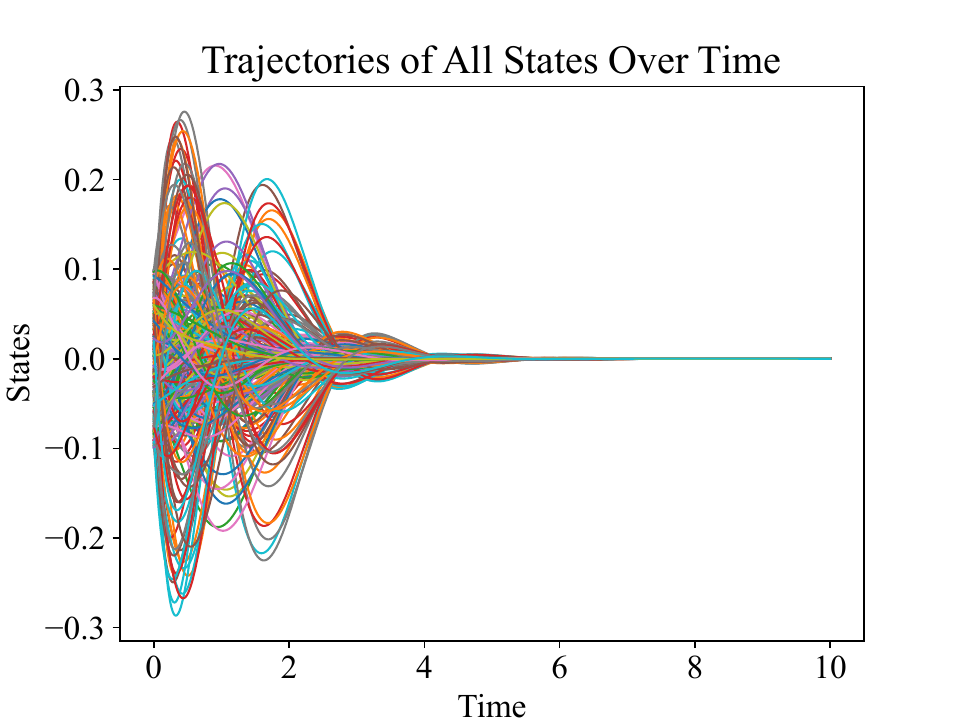}}
		\subfigure[Cost error of 3D quadrotor]{
			\label{Fig2.sub3}
			\includegraphics[width=.35\textwidth]{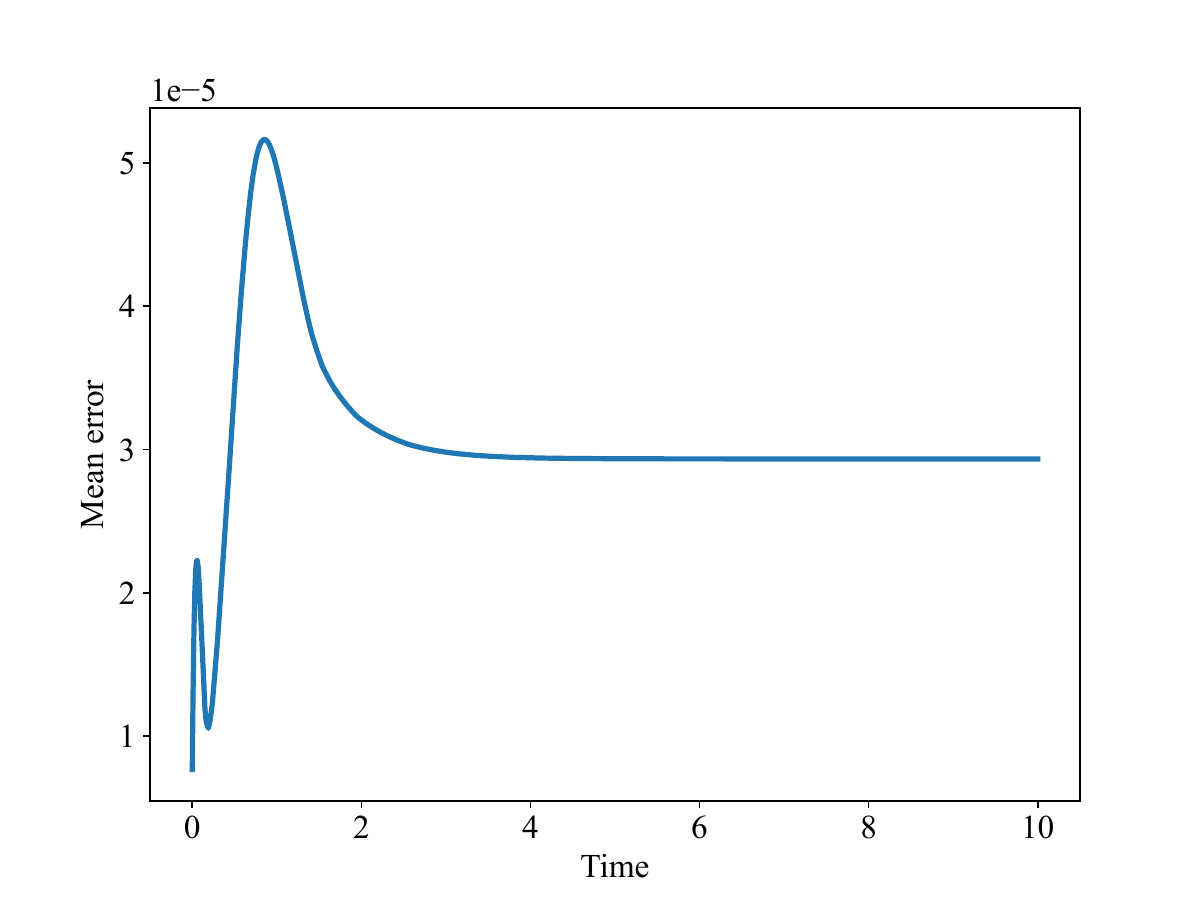}}
		\subfigure[Trajectories for 3D quadrotor]{
			\label{Fig2.sub4}
			\includegraphics[width=.35\textwidth]{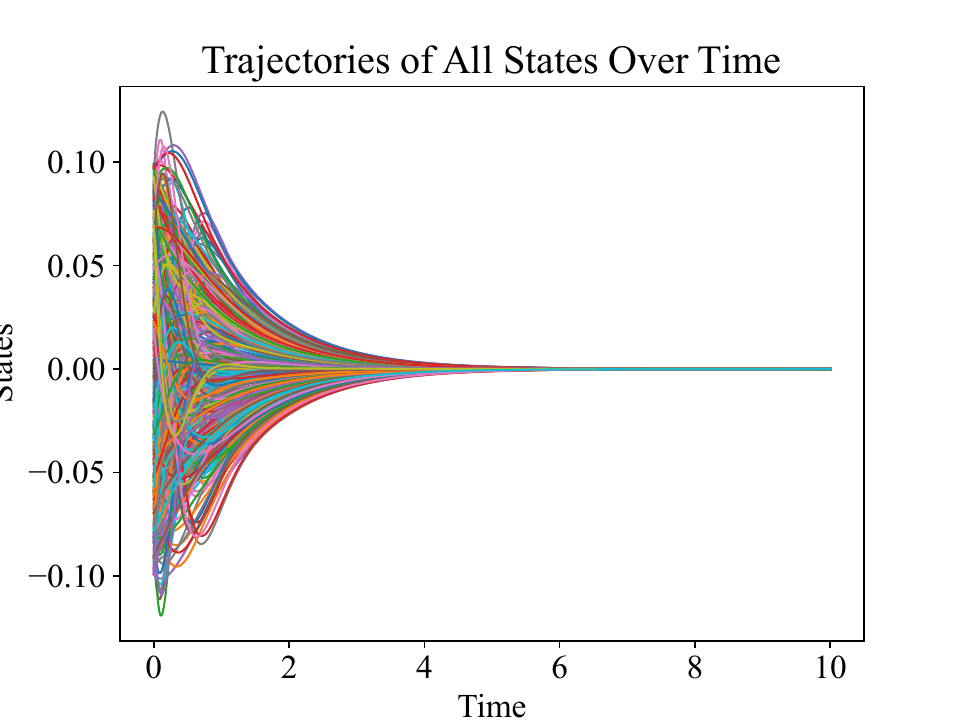}}
		\caption{The error between accumulated costs computed by the control learned from identified system and the true system (left). 50 trajectories of all states driven by the control learned from the identified system (right).}
		\label{fig2}
	\end{figure}
    
\section{Conclusion}

We proposed a novel approach for solving optimal control problems in high-dimensional nonlinear systems. The results demonstrated the effectiveness of our method in achieving stabilizing control and accurately approximating value functions, even for systems with state dimensions up to 9 and input dimensions up to 4. However, we acknowledge that the region of consideration for high-dimensional systems in this study is relatively limited. This reflects the inherent challenges in solving high-dimensional HJB equations, including computational complexity and sensitivity to identification errors. Addressing these limitations will be a key focus of our future research, aiming to broaden the stabilized region of the learned controls.
	
	\newpage
	\bibliography{cite}
	\appendix
    \section{Proof of Theorem \ref{conv_pi}}\label{sec:appendixA}
    \begin{proof}
		To prove the statement, it suffices to show that $\forall \theta>0,$ there exists $\delta>0$ such that if 
		$\|\kappa_{h}^{(i-1)}-\kappa^{(i-1)}\|_\infty <\delta, \|f_h-f\|_\infty <\delta, \|g_h-g\|_\infty <\delta$
		, we have $\|V^{(i)}_h-V^{(i)}\|_\infty <\theta, \|\kappa^{(i)} - \kappa_{h}^{(i)}\|_\infty  < \theta$. 
        
		{\bf Step 1 (Value functions):} In the $i$-th iteration of policy evaluation, the value functions $V_h^{i}$ and $V^{i}$ are solutions to the GHJB equations of the identified system and true system \begin{align}
			&\text{(GHJB)$_h$}\quad DV^{(i)}_h \cdot F_h^{(i)} = -L_h^{(i)},\\
			&\text{(GHJB)}\quad DV^{(i)} \cdot F^{(i)} = -L^{(i)},
		\end{align}
		where $L_h^{(i)}(\bx) = -Q(\bx)-\kappa_h^{(i-1)}(\bx)^T R\kappa_h^{(i-1)}(\bx), L^{(i)}(\bx) = -Q(\bx)-\kappa^{(i-1)}(\bx)^T R\kappa^{(i-1)}(\bx)$. The solutions are \begin{align}
			&V_h^{(i)}(\bx) = \int_{0}^{\infty} L_h^{(i)}(S_h^{(i)}(t,\bx)){\rm d}t,\\
			&V^{(i)}(\bx) = \int_{0}^{\infty} L^{(i)}(S^{(i)}(t,\bx)){\rm d}t,
		\end{align} where $S_h^{(i)}(t,\bx)$ and $S^{(i)}(t,\bx)$ are solutions to $\dot{\bx} = F_h^{(i)}(\bx,\kappa^{(i-1)}_h(\bx))$, and $\dot{\bx} = F^{(i)}(\bx,\kappa^{(i-1)}(\bx)),$ respectively. Since the value functions $V_h^{(i)}$ and $V^{(i)}$ are actually Lyapunov functions of these systems \citep{liu2023physics}, we have $V_h^{(i)}(\bx)<\infty$ and $V^{(i)}(\bx)<\infty$ for $\bx \in \Omega$. Hence, for $\forall \theta>0$ there exists a sufficiently large $T_1>0$ such that \begin{align}\label{secondbound}
		&0\le \int_{T_1}^{\infty} L_h^{(i)}(S_h^{(i)}(t,\bx)){\rm d}t\le \theta/2,\\
		&0\le \int_{T_1}^{\infty} L^{(i)}(S^{(i)}(t,\bx)){\rm d}t\le \theta/2.
		\end{align}
		It follows that \begin{equation}\label{key1}
			-\theta/2\le\int_{T_1}^{\infty} L_h^{(i)}(S_h^{(i)}(t,\bx)){\rm d}t- \int_{T_1}^{\infty} L^{(i)}(S^{(i)}(t,\bx)){\rm d}t \le \theta/2.
		\end{equation}
		Then we consider \begin{align}
			&\left\|\int_{0}^{T_1} L_h^{(i)}(S_h^{(i)}(t,\bx))-  L^{(i)}(S^{(i)}(t,\bx)){\rm d}t\right\|_\infty  \le \int_{0}^{T_1} \left\|L_h^{(i)}(S_h^{(i)}(t,\bx))-  L^{(i)}(S^{(i)}(t,\bx))\right\|_\infty {\rm d}t\\
			\le &\int_{0}^{T_1} \left\|L_h^{(i)}(S_h^{(i)}(t,\bx))- L^{(i)}(S_h^{(i)}(t,\bx))\right\|_\infty {\rm d}t+\int_{0}^{T_1} \left\|L^{(i)}(S_h^{(i)}(t,\bx))- L^{(i)}(S^{(i)}(t,\bx))\right\|_\infty {\rm d}t\\
			\le& T_1\|L_h^{(i)}-L^{(i)}\|_\infty + \text{Lip}_L\int_{0}^{T_1}\left\|S_h^{(i)}(t,\bx) -S^{(i)}(t,\bx) \right\|_\infty {\rm d}t,\label{first}
		\end{align} where $\text{Lip}_L$ is the maximum Lipschitz constant among that of $L^{(i)}$. There exists $C_{\kappa}>0$ such that \begin{align}
            \|L^{(i)}_h-L^{(i)}\|_\infty  = \|(\kappa^{(i-1)})^T R\kappa^{(i-1)} - (\kappa_h^{(i-1)})^T R\kappa_h^{(i-1)}\|_\infty \le C_{\kappa} \|\kappa^{(i-1)} - \kappa_h^{(i-1)}\|_\infty .\label{Lbound}
        \end{align}
		Based on continuous dependence of dynamical system \citep[Theorem 3.4, Chapter 3]{khalil2002nonlinear}, we have \begin{equation}\label{Serror}
			\left\|S^t_h(\bx) -S^t(\bx) \right\|_\infty \le \frac{\|F_h^{(i)}-F^{(i)}\|_\infty }{\text{Lip}_{F}}\left(\exp(\text{Lip}_{F} t)-1\right),
		\end{equation} where $\text{Lip}_{F}$ is the maximum Lipschitz constant among that of $F^{(i)}$. For $\|F_h^{(i)}-F^{(i)}\|$, we have \begin{align}
			\|F_h^{(i)}-F^{(i)}\|_\infty  &\le \|f_h-f\|_\infty  + \|g_h\kappa_h^{(i-1)}-g\kappa_h^{(i-1)}\|_\infty +\|g\kappa_h^{(i-1)}-g\kappa^{(i-1)}\|_\infty \\
			& \le \|f_h-f\|_\infty  + \|g_h-g\|_\infty \|\kappa_h^{(i-1)}\|_\infty +\|g\|_\infty \|\kappa_h^{(i-1)}-\kappa^{(i-1)}\|_\infty .
            \label{Ferror}
		\end{align}
		Substituting Eq. \eqref{Lbound}, Eq. \eqref{Serror} and Eq. \eqref{Ferror} into Eq. \eqref{first}, we have \begin{align}
			&\left\|\int_{0}^{T_1} L_h^{(i)}(S_h^{(i)}(t,\bx))-  L^{(i)}(S^{(i)}(t,\bx)){\rm d}t\right\|_\infty \\
			\le&T_1\left(\|L^{(i)}_h-L^{(i)}\|_\infty  + \text{Lip}_L\frac{\|F_h^{(i)}-F^{(i)}\|_\infty }{\text{Lip}_{F}}\left(\exp(\text{Lip}_{F} T_1)-1\right)\right)\\
            \le& C_{V_1} \|\kappa_h^{(i-1)}-\kappa^{(i-1)}\|_\infty  + C_{V_2}\|f_h-f\|_\infty  + C_{V_3}\|g_h - g\|_\infty ,
		\end{align} where $C_{V_1} = T_1 C_\kappa + T_1\|g\|_\infty \frac{\text{Lip}_L}{\text{Lip}_{F}}\left(\exp(\text{Lip}_{F} T_1)-1\right), C_{V_2} = T_1\frac{\text{Lip}_L}{\text{Lip}_{F}}\left(\exp(\text{Lip}_{F} T_1)-1\right), $ and $C_{V_3} = T_1\|\kappa_h^{(i-1)}\|_\infty \frac{\text{Lip}_L}{\text{Lip}_{F}}\left(\exp(\text{Lip}_{F} T_1)-1\right)$. Then for $\forall \theta>0,$ there exists $\delta_1$ such that, if $\|\kappa_{h}^{(i-1)}-\kappa^{(i-1)}\|_\infty <\delta_1, \|f_h-f\|_\infty <\delta_1, \|g_h-g\|_\infty <\delta_1$
		, we have \begin{equation}\label{key2}\left\|\int_{0}^{T_1} L_h^{(i)}(S_h^{(i)}(t,\bx))-  L^{(i)}(S^{(i)}(t,\bx)){\rm d}t\right\|_\infty \le \theta/2.
		\end{equation}
		Combining \eqref{key1} and \eqref{key2}, we have \begin{align}
			\|V_h^{(i)}(\bx)-V^{(i)}(\bx)\|_\infty  \le \left\|\int_{0}^{T_1} \left(L_h^{(i)}(S_h^t(\bx)) - L^{(i)}(S^t(\bx))\right) {\rm d}t\right\|_\infty  + \\
			 \left\|\int_{T_1}^{\infty} \left(L_h^{(i)}(S_h^t(\bx)) - L^{(i)}(S^t(\bx))\right) {\rm d}t\right\|_\infty  \le \theta.
		\end{align} 
		
		{\bf Step 2 (Control laws):} The optimal control laws are updated as \begin{align}
			\kappa^{(i)}(\bx) = -\frac{1}{2}R^{-1}g^T(\bx)(DV^{(i)})^T,~ \kappa_h^{(i)}(\bx) =  -\frac{1}{2}R^{-1}g_h^T(\bx)(DV_h^{(i)})^T.
		\end{align} 
		The error can be expressed as\begin{equation}
			\kappa^{(i)}(\bx) - \kappa_h^{(i)}(\bx) = -\frac{1}{2}R^{-1}\left(g^T(\bx)(DV^{(i)}(\bx))^T-g_h^T(\bx)(DV_h^{(i)}(\bx))^T\right).
		\end{equation}
		It follows that\begin{align}
			\|\kappa^{(i)}(\bx) - \kappa_h^{(i)}(\bx)\|_\infty &\le \frac{1}{2}\|R^{-1}\|  \| g^T(DV^{(i)})^T-g_h^T(DV_h^{(i)})^T \|_\infty \\
			&\le \frac{1}{2}\|R^{-1}\|\left(\|g^T-g_h^T\|_\infty \|DV^{(i)}\|_\infty  + \|g_h^T\|_\infty \|DV^{(i)} - DV_h^{(i)}\|_\infty \right).\label{kappabound}
		\end{align}
		
		It suffices to bound each term on the r.h.s. of \eqref{kappabound}. Since $F_h^{(i)}$, $F^{(i)}\in C^1(\Omega)$, and $L^{(i)}, L^{(i)}_h\in C^1(\Omega)$, the value functions $V^{(i)}$ and $V^{(i)}_h$ are continuously differentiable and given by \citep[Proposition 2]{liu2023physics} \begin{align}
		DV^{(i)}(\bx) = \int_{0}^{\infty}DL^{(i)}(S^{(i)}(t,\bx))\Phi^{(i)}(t,\bx){\rm d}t,\label{DV}\\
		DV_h^{(i)}(\bx) = \int_{0}^{\infty}DL^{(i)}_h(S_h^{(i)}(t,\bx))\Phi_h^{(i)}(t,\bx){\rm d}t,\label{DVh}
		\end{align} where $\Phi^{(i)}(t,\bx)$ and $\Phi^{(i)}_h(t,\bx)$ are the fundamental matrix solutions of the initial value problems: \begin{align}
		&\dot{\Phi}(t,\bx) = A^{(i)}(t,\bx)\Phi(t,\bx), ~\Phi(0) = I, ~A^{(i)}(t,\bx) = DF^{(i)}(\bx),\\
		&\dot{\Phi}_h(t,\bx) = A^{(i)}_h(t,\bx)\Phi_h(t,\bx), ~\Phi_h(0) = I, ~A^{(i)}_h(t,\bx) = DF_h^{(i)}(\bx).
		\end{align} 
		Thus, there exists a constant $C_1>0$ such that $\|DV^{(i)}\|_\infty \le C_1.$ Combining these results, we have \begin{equation}\label{bound1}
			\|g^T-g_h^T\|_\infty \|DV^{(i)}\|_\infty \le \|g^T-g_h^T\|_\infty  C_1.
		\end{equation}
		Then we consider \begin{align}\label{DVerror}
			DV^{(i)} - DV^{(i)}_h = \int_{0}^{\infty}\left(DL^{(i)}(S^{(i)}(t,\bx))\Phi^{(i)}(t,\bx) - DL_h^{(i)}(S_h^{(i)}(t,\bx))\Phi^{(i)}_h(t,\bx)\right){\rm d}t.
		\end{align}
		Expanding the difference   in the integrand, we get \begin{equation}\label{goalbound}
			\begin{aligned}
			&DL^{(i)}(S^{(i)}(t,\bx))\Phi^{(i)}(t,\bx) - DL_h^{(i)}(S_h^{(i)}(t,\bx))\Phi^{(i)}_h(t,\bx) \\
			=& DL^{(i)}(S^{(i)}(t,\bx))\Phi^{(i)}(t,\bx) - DL_h^{(i)}(S^{(i)}(t,\bx))\Phi^{(i)}(t,\bx)\\
			+&DL_h^{(i)}(S^{(i)}(t,\bx))\Phi^{(i)}(t,\bx) - DL_h^{(i)}(S_h^{(i)}(t,\bx))\Phi^{(i)}(t,\bx)\\ +&DL_h^{(i)}(S_h^{(i)}(t,\bx))\Phi^{(i)}(t,\bx) - DL_h^{(i)}(S_h^{(i)}(t,\bx))\Phi^{(i)}_h(t,\bx).
			\end{aligned}
		\end{equation}
		Then we proceed term by term to derive the bound. For the first term, since $L^{(i)}_h, L^{(i)}\in C^1$, there exists $C_{L}>0$, such that \begin{align}
			&\|DL^{(i)}(S^{(i)}(t,\bx))\Phi^{(i)}(t,\bx) - DL_h^{(i)}(S^{(i)}(t,\bx))\Phi^{(i)}(t,\bx)\|_\infty  \le C_{L}\|L^{(i)} - L_h^{(i)}\|_\infty \|\Phi^{(i)}(t,\bx)\|_\infty .
		\end{align}
        Given Eq. \eqref{Lbound} and the fact that $\|\Phi(t,\bx)\|_\infty $ is bounded by $C_\Phi>0$, we have \begin{equation} \label{bound2}
			\|DL^{(i)}\Phi^{(i)} - DL_h^{(i)}\Phi^{(i)}\|_\infty \le C_{L}C_{\Phi} C_{\kappa} \|\kappa^{(i-1)} - \kappa_h^{(i-1)}\|_\infty .
		\end{equation} 
		
		For the second term of \eqref{goalbound}, assuming $DL_h$ is Lipschitz with respect to $\bx \in \Omega,$ there exists $\text{Lip}_{DL}>0,$ such that \begin{align}
			&\|DL_h^{(i)}(S^{(i)}(t,\bx))\Phi^{(i)}(t,\bx) - DL_h^{(i)}(S_h^{(i)}(t,\bx))\Phi^{(i)}(t,\bx)\|_\infty \\
			\le& \text{Lip}_{DL_h}\|S^{(i)}(t,\bx) - S_h^{(i)}(t,\bx)\|_\infty \|\Phi^{(i)}(t,\bx)\|_\infty \\
			\le& \text{Lip}_{DL_h}C_{\Phi}\frac{\|F_h^{(i)}-F^{(i)}\|_\infty }{\text{Lip}_{F}}\left(\exp(\text{Lip}_{F} t)-1\right).\label{bound3}
		\end{align}
		
		The third term involves $\Phi^{(i)}(t,\bx) - \Phi^{(i)}(t,\bx)$ that satisfies:\begin{align}
			\frac{\rm d}{{\rm d}t}(\Phi^{(i)}(t,\bx) - \Phi^{(i)}(t,\bx))  = A^{(i)}(t,\bx)\Phi^{(i)}(t,\bx)- A_h^{(i)}(t,\bx)\Phi_h^{(i)}(t,\bx)\\
			=(A^{(i)} - A_h^{(i)})\Phi^{(i)}(t,\bx)+ A_h^{(i)} (\Phi^{(i)}(t,\bx) - \Phi_h^{(i)}(t,\bx))
		\end{align} By Gronwall's inequality, we have  \begin{align}
			\|\Phi^{(i)}(t,\bx) - \Phi_h^{(i)}(t,\bx)\|_\infty \le \int_{0}^{t} \|A^{(i)} - A_h^{(i)}\|_\infty \|\Phi^{(i)}(s,\bx)\|_\infty  e^{\int_{s}^{t}\|A_h^{(i)}(\tau,\bx)\|_\infty {\rm d}\tau}{\rm d}s\label{phibound}
		\end{align}
		Using the fact that $F_h^{(i)},F^{(i)}\in C^1(\Omega)$, there exists $C_F>0$ such that \begin{align}
			\|A^{(i)}(t,\bx) - A^{(i)}_h(t,\bx)\|_\infty  = \|DF^{(i)} -DF^{(i)}_h\|_\infty  \le C_F \|F^{(i)} -F^{(i)}_h\|_\infty .
		\end{align} 
		If $\|A_h^{(i)}(s,\bx)\|_\infty $ is bounded by $C_{A_h}>0$, it follows from \eqref{phibound} that \begin{equation}
			\|\Phi^{(i)}(t,\bx) - \Phi_h^{(i)}(t,\bx)\|_\infty  \le t C_F C_\Phi\|F^{(i)} - F_h^{(i)}\|_\infty  e^{C_{A_h} t}.
		\end{equation} 
		Since $L_h^{(i)}\in C^1(\Omega),$ there exists $C_{L_h}>0$ such that $\|DL^{(i)}_h(S^{(i)}(t,\bx))\|_\infty \le C_{L_h}$.
		It follows that the third term of \eqref{goalbound} is bounded by \begin{align}
			\|DL_h^{(i)}(S_h^{(i)}(t,\bx))\Phi^{(i)}(t,\bx) - DL_h^{(i)}(S_h^{(i)}(t,\bx))\Phi^{(i)}_h(t,\bx)\|_\infty \le te^{C_A t}C_{L_h}C_FC_\Phi\|F^{(i)} - F^{(i)}_h\|_\infty .\label{bound4}
		\end{align}
		For the integral in \eqref{DVerror}, we can found a sufficiently $T_2>0$ such that \begin{equation}
			\frac{1}{2}\|R^{-1}\|\left(\left\|\int_{T_2}^{\infty} DL^{(i)}(S^{(i)}(t,\bx))\Phi^{(i)}(t,\bx){\rm d}t\right\|_\infty  + \left\|\int_{T_2}^{\infty}  DL_h^{(i)}(S_h^{(i)}(t,\bx))\Phi^{(i)}_h(t,\bx) {\rm d}t\right\|_\infty\right)\le \theta/2.
		\end{equation}
		Combining \eqref{bound1}, \eqref{bound2}, \eqref{bound3}, and \eqref{bound4}, we can write \eqref{kappabound} as \begin{equation}
			\begin{aligned}
				&\|\kappa^{(i)} - \kappa_h^{(i)}\|_\infty \le \frac{1}{2}\|R^{-1}\|\left(C_1\|g^T-g_h^T\|_\infty+  C_2 \|\kappa^{(i-1)} - \kappa_h^{(i-1)}\|_\infty+  (C_3+C_4)\|F^{(i)} - F^{(i)}_h\|_\infty \right)\\
				+&\frac{1}{2}\|R^{-1}\|\left(\left\|\int_{T_2}^{\infty} DL^{(i)}(S^{(i)}(t,\bx))\Phi^{(i)}(t,\bx){\rm d}t\right\|_\infty + \left\|\int_{T_2}^{\infty}  DL_h^{(i)}(S_h^{(i)}(t,\bx))\Phi^{(i)}_h(t,\bx){\rm d}t\right\|_\infty\right),
			\end{aligned}
		\end{equation} where \begin{align}
		    &C_2 = T_2\|g_h^T\|_\infty C_{L}C_{\Phi} C_{\kappa}\\
            &C_3 = T_2\|g_h^T\|_\infty C_{\Phi}\frac{\text{Lip}_{DL_h}}{\text{Lip}_{F}}\left(\exp(\text{Lip}_{F} T_2)-1\right),\\
            &C_4=T_2 \|g_h^T\|_\infty T_2e^{C_A T_2}C_{L_h}C_FC_\Phi.
		\end{align}
 		Similar to the previous proof, we can find $\delta_2>0$, such that  $\|\kappa^{(i)}(\bx) - \kappa_h^{(i)}(\bx)\|_\infty\le \theta$. Thus, the proof is completed by setting $\delta = \min\{\delta_1,\delta_2\}$.
	\end{proof}

\end{document}